\documentclass[journal]{IEEEtran}

\IEEEoverridecommandlockouts
\usepackage{cite}
\usepackage{amsmath,amssymb,amsfonts}
\usepackage{graphicx}
\usepackage{acronym,comment}
\usepackage[colorinlistoftodos]{todonotes}

\usepackage{textcomp}
\usepackage{xcolor}

\usepackage{bbm}
\usepackage{amsthm}
\usepackage{psfrag,caption,subcaption}

\usepackage{bbm}
\usepackage{algorithm}
\usepackage[noend]{algpseudocode} 

\algrenewcommand\algorithmiccomment[1]{// {\itshape #1}}

\newtheorem{mylem}{Lemma}
\newtheorem{myrem}{Remark}%

\usepackage[utf8]{inputenc}
\usepackage[english]{babel}
\newtheorem{theorem}{Theorem}
\def\cast{{
   \mathord{
      \hbox to 0em{
         \ooalign{
	   \smash{\hbox{$\ast$}}\crcr
	   \smash{\hskip-1pt\Large\hbox{$\circ$}} }
	 \hidewidth}
      \phantom{\bigcirc}
} }}

\newcommand{\rH}{^{ \raisebox{1pt}{$\rm \scriptscriptstyle H$}}}
\newcommand{\rT}{^{ \raisebox{1.2pt}{$\rm \scriptstyle T$}}}

\newcommand{\bds}{\begin {itemize}}
\newcommand{\eds}{\end {itemize}}
\newcommand{\bdf}{\begin{definition}}
\newcommand{\blm}{\begin{lemma}}
\newcommand{\edf}{\end{definition}}
\newcommand{\elm}{\end{lemma}}
\newcommand{\bthm}{\begin{theorem}}
\newcommand{\ethm}{\end{theorem}}
\newcommand{\bprp}{\begin{prop}}
\newcommand{\eprp}{\end{prop}}
\newcommand{\bcl}{\begin{claim}}
\newcommand{\ecl}{\end{claim}}
\newcommand{\bcr}{\begin{coro}}
\newcommand{\ecr}{\end{coro}}
\newcommand{\bquest}{\begin{question}}
\newcommand{\equest}{\end{question}}

\newcommand{\larrow}{{\larrow}}

\newcommand{\argmin}{\ensuremath{\mathrm{arg}\min}}
\newcommand{\argmax}{\ensuremath{\mathrm{arg}\max}}

\newcommand{\cA}{{\ensuremath{\mathcal{A}}}}

\newcommand{\cC}{{\ensuremath{\mathcal{C}}}}

\newcommand{\cN}{{\ensuremath{\mathcal{N}}}}
\newcommand{\cO}{{\ensuremath{\mathcal{O}}}}
\newcommand{\cP}{{\ensuremath{\mathcal{P}}}}

\def\mbC{{\ensuremath{\mathbb C}}}

\def\mbE{{\ensuremath{\mathbb E}}}

\newcommand{\va}{{\ensuremath{{\mathbf{a}}}}}

\newcommand{\vg}{{\ensuremath{{\mathbf{g}}}}}

\newcommand{\vh}{{\ensuremath{{\mathbf{h}}}}}

\newcommand{\vp}{{\ensuremath{{\mathbf{p}}}}}
\newcommand{\vq}{{\ensuremath{{\mathbf{q}}}}}

\newcommand{\vu}{{\ensuremath{{\mathbf{u}}}}}
\newcommand{\vv}{{\ensuremath{{\mathbf{v}}}}}

\newcommand{\vx}{{\ensuremath{{\mathbf{x}}}}}

\newcommand{\mB}{{\ensuremath{\mathbf{B}}}}

\newcommand{\mH}{{\ensuremath{\mathbf{H}}}}

\newcommand{\mI}{{\ensuremath{\mathbf{I}}}}

\newcommand{\mJ}{{\ensuremath{\mathbf{J}}}}

\newcommand{\mV}{{\ensuremath{\mathbf{V}}}}

\usepackage{latexsym}

\def\IC{\mathbb C}
\def\IN{\mathbb N}
\def\IZ{\mathbb Z}
\def\IR{\mathbb R}

\def\shat{^{\mathchoice{}{}%
 {\,\,\smash{\hbox{\lower4pt\hbox{$\widehat{\null}$}}}}%
 {\,\smash{\hbox{\lower3pt\hbox{$\hat{\null}$}}}}}}

\def\bSigma{{
      \ooalign{
      \smash{\hskip.4pt\raise.4pt\hbox{$\Sigma$}}\vphantom{}\crcr
      \smash{\hskip.7pt\raise.6pt\hbox{$\Sigma$}}\vphantom{}\crcr
      \smash{\hbox{$\Sigma$}}\vphantom{$\Sigma$}}
      \vphantom{\hbox{$\Sigma$}}
      }}
\def\bTheta{{
      \ooalign{
      \smash{\hskip.5pt\raise.5pt\hbox{$\Theta$}}\vphantom{}\crcr
      \smash{\hskip.0pt\raise.1pt\hbox{$\Theta$}}\vphantom{}\crcr
      \smash{\hbox{$\Theta$}}\vphantom{$\Theta$}}
      \vphantom{\hbox{$\Theta$}}
      }}
\def\bDelta{{
      \ooalign{
      \smash{\hskip.4pt\raise.4pt\hbox{$\Delta$}}\vphantom{}\crcr
      \smash{\hskip.7pt\raise.6pt\hbox{$\Delta$}}\vphantom{}\crcr
      \smash{\hbox{$\Delta$}}\vphantom{$\Delta$}}
      \vphantom{\hbox{$\Delta$}}
      }}
\def\bLambda{{
      \ooalign{
      \smash{\hskip.5pt\raise.5pt\hbox{$\Lambda$}}\vphantom{}\crcr
      \smash{\hskip.0pt\raise.1pt\hbox{$\Lambda$}}\vphantom{}\crcr
      \smash{\hbox{$\Lambda$}}\vphantom{$\Lambda$}}
      \vphantom{\hbox{$\Lambda$}}
      }}

\makeatletter

\def\bordermatrix#1{\begingroup \m@th
  \@tempdima 8.75\p@
  \setbox\z@\vbox{%
    \def\cr{\crcr\noalign{\kern2\p@\global\let\cr\endline}}%
    \ialign{$##$\hfil\kern2\p@\kern\@tempdima&\thinspace\hfil$##$\hfil
      &&\quad\hfil$##$\hfil\crcr
      \omit\strut\hfil\crcr\noalign{\kern-\baselineskip}%
      #1\crcr\omit\strut\cr}}%
  \setbox\tw@\vbox{\unvcopy\z@\global\setbox\@ne\lastbox}%
  \setbox\tw@\hbox{\unhbox\@ne\unskip\global\setbox\@ne\lastbox}%
  \setbox\tw@\hbox{$\kern\wd\@ne\kern-\@tempdima\left[\kern-\wd\@ne
    \global\setbox\@ne\vbox{\box\@ne\kern2\p@}%
    \vcenter{\kern-\ht\@ne\unvbox\z@\kern-\baselineskip}\,\right]$}%
  \null\;\vbox{\kern\ht\@ne\box\tw@}\endgroup}
\makeatother

\makeatletter
\def\argmin{\mathop{\operator@font arg\,min}}
\def\argmax{\mathop{\operator@font arg\,max}}
\makeatother

\newcommand{\bea}{\begin{array}}
\newcommand{\ena}{\end{array}}
\newcommand{\beq}{\begin{equation}}
\newcommand{\enq}{\end{equation}}

\newcommand{\beqa}{\begin{eqnarray}}
\newcommand{\enqa}{\end{eqnarray}}

\newcommand{\beqan}{\begin{eqnarray*}}
\newcommand{\enqan}{\end{eqnarray*}}

\newcommand{\AL}{\begin{enumerate}}
\newcommand{\ALE}{\end{enumerate}}

\def\addots{\mathinner{
    \mkern1mu\raise0pt\vbox{\kern7pt\hbox{.}}
    \mkern2mu\raise4pt\hbox{.}
    \mkern2mu\raise7pt\hbox{.}
    \mkern1mu}}

\def\sddots{\mathinner{
    \mkern.8mu\raise7pt\hbox{.}
    \mkern.8mu\raise4pt\hbox{.}
    \mkern.8mu\raise0pt\vbox{\kern7pt\hbox{.}}
    \mkern1mu}}

\def\saddots{\mathinner{
    \mkern.2mu\raise0pt\vbox{\kern7pt\hbox{.}}
    \mkern.2mu\raise4pt\hbox{.}
    \mkern.2mu\raise7pt\hbox{.}
    \mkern1mu}}

\def\sqplus{\mathbin{
	{\ooalign{\hfil\raise.3ex\hbox{\scriptsize
	+}\hfil\crcr\mathhexbox274\crcr\mathhexbox275}}
	}} 
\def\sqminus{\mathbin{
	{\ooalign{\hfil\raise.3ex\hbox{\scriptsize
	--}\hfil\crcr\mathhexbox274\crcr\mathhexbox275}}
	}}

\def\IC{{
   \mathord{
      \hbox to 0em{
	 \hskip-4pt
         \ooalign{
	   \smash{\hskip1.9pt\raise2.6pt\hbox{$\scriptscriptstyle |$}}\crcr
	   \smash{\hbox{\rm\sf C}} }
	 \hidewidth}
      \phantom{\hbox{\rm\sf C}}
} }}
\def\IN{
    {\ooalign{
   \smash{\hskip2.2pt\raise1.5pt\hbox{$\scriptscriptstyle |$}}\vphantom{}\crcr
   \hbox{\sf N}
	}}
	} 
\def\IZ{
    {\ooalign{
   \smash{\hskip1.9pt\raise0pt\hbox{$\sf Z$}}\vphantom{}\crcr
   \hbox{\sf Z}
	}}
	} 
\def\IR{
    {\ooalign{
   \smash{\hskip2.2pt\raise1.5pt\hbox{$\scriptscriptstyle |$}}\vphantom{}\crcr
   \smash{\hskip2.2pt\raise3.3pt\hbox{$\scriptscriptstyle |$}}\vphantom{}\crcr
   \hbox{\sf R}
	}}
	} 

\DeclareMathAlphabet{\mathcmb}{OT1}{cmr}{b}{n}

\def\bSigma{\ensuremath{\mathcmb{\Sigma}}}
\def\bLambda{\ensuremath{\mathcmb{\Lambda}}}

\def\bTheta{\ensuremath{\mathcmb{\Theta}}}

\newcommand{\SI}{\begin{indlist}}
\newcommand{\EI}{\end{indlist}}

\newcommand{\DL}{\begin{dashlist}}
\newcommand{\DLE}{\end{dashlist}}

\makeatletter
\def\setboxz@h{\setbox\z@\hbox}
\def\wdz@{\wd\z@}
\def\boxz@{\box\z@}
\def\underset#1#2{\binrel@{#2}%
  \binrel@@{\mathop{\kern\z@#2}\limits_{#1}}}
\def\binrel@#1{\begingroup
  \setboxz@h{\thinmuskip0mu
    \medmuskip\m@ne mu\thickmuskip\@ne mu
    \setbox\tw@\hbox{$#1\m@th$}\kern-\wd\tw@
    ${}#1{}\m@th$}%
  \edef\@tempa{\endgroup\let\noexpand\binrel@@
    \ifdim\wdz@<\z@ \mathbin
    \else\ifdim\wdz@>\z@ \mathrel
    \else \relax\fi\fi}%
  \@tempa
}
\let\binrel@@\relax%
\makeatother

\acrodef{ris}[RIS]{reconfigurable intelligent surface}
\acrodef{isac}[ISAC]{integrated sensing and communication}
\acrodef{dfbs}[DFBS]{Dual function radar communication base station}
\acrodef{ue}[UE]{user equipment}
\acrodef{sinr}[SINR]{signal-to-interference-plus-noise ratio}
\acrodef{snr}[SNR]{signal-to-noise ratio}
\acrodef{mui}[MUI]{multi-user-interference}
\acrodef{mimo}[MIMO]{multiple-input-multiple-output}
\acrodef{miso}[MISO]{multiple-input-single-output}
\acrodef{ura}[URA]{uniform rectangular array}
\acrodef{nlos}[NLoS]{non-line-of-sight}
\acrodef{los}[LoS]{line-of-sight}
\acrodef{wrt}[w.r.t.]{with respect to}
\acrodef{rcc}[RCC]{radar-communication-coexistence}
\acrodef{ula}[ULA]{uniform linear array}
\acrodef{bs}[BS]{base station}
\acrodef{cdf}[c.d.f.]{cumulative distribution function}
\acrodef{sdp}[SDP]{semi-definite program}
\acrodef{sdr}[SDR]{semi-definite relaxation}
\acrodef{qos}[QoS]{quality of service}
\acrodef{iid}[IID]{independent and identically distributed}
\acrodef{dris}[DP-HRIS]{dynamically programmable hybrid reconfigurable intelligent surface}
\acrodef{hris}[HRIS]{hybrid reconfigurable intelligent surface}

\acrodef{pris}[PRIS]{passive RIS}
\acrodef{aris}[ARIS]{active RIS}

\acrodef{tdm}[TDM]{time division multiplexing}

\def\BibTeX{{\rm B\kern-.05em{\sc i\kern-.025em b}\kern-.08em
		T\kern-.1667em\lower.7ex\hbox{E}\kern-.125emX}}
\begin{document}

	\title{		
Optimal Placement of Active and Passive Elements in Hybrid RIS-assisted Communication Systems
	}
	
	\author{\IEEEauthorblockN{R.S. Prasobh Sankar,~\IEEEmembership{Student Member, IEEE}, and Sundeep Prabhakar Chepuri,~\IEEEmembership{Member, IEEE} \thanks{ The authors are with the Department of Electrical Communication Engineering, Indian Institute of Science, Bengaluru, India. Email:\{prasobhr,spchepuri\}@iisc.ac.in. This work is supported in part by the Next Generation Wireless Research and Standardization on 5G and Beyond project, MeitY and PMRF, Government of India.} }

	}
	\maketitle
	
	\begin{abstract}
	
		Hybrid reconfigurable intelligent surfaces (HRIS) are RIS architectures having both active and passive elements. The received signal-to-noise ratio (SNR) in HRIS-assisted communication systems depends on the placement of active elements. In this paper, we show that received SNR can be improved with a channel-aware  placement of the active elements. We jointly design the transmit precoder, the RIS coefficients, and the location of active and passive elements of the HRIS to maximize the SNR. We solve the underlying combinatorial nonconvex optimization problem using alternating optimization and propose a low-complexity solver, which is provably nearly optimal. Through numerical simulations, we demonstrate that the proposed method offers significantly improved performance compared to communication systems having a fully passive RIS array or a hybrid RIS array with channel agnostic active element placement and performance comparable to that of communication systems with a fully active RIS array.

	\end{abstract} 
	
	
\section{Introduction}
\Acp{ris} are arrays with tunable elements that can be remotely controlled to favorably modify the wireless propagation environment and introduce additional paths, which are crucial for next-generation wireless systems operating in millimeter and beyond frequency bands with harsh propagation environments~\cite{basar2019wireless}.

\acp{ris}, as originally envisioned, are fully passive arrays comprising passive phase shifters. One of the main limitations of  \acp{ris} is the so-called \textit{double pathloss} effect, i.e., the pathloss of the link via the \ac{ris} is the product of the pathlosses of the link between the \ac{ris} and transmitter, e.g., a \ac{bs}, and between the \ac{ris} and receiver, e.g., a~\ac{ue}. Due to the double pathloss, the performance gains offered by a fully \ac{pris} are often insignificant, especially when a direct path already exists between the communicating terminals. To circumvent the limitations of \acp{pris}, a different \ac{ris} architecture known as the \ac{aris} has been proposed~\cite{zhang2022active}. Unlike \acp{pris}, \acp{aris} have active elements that can phase shift as well as amplify the incident signal. Active elements comprise passive phase shifters and reflection-type amplifiers, which consume relatively less power than decode-and-forward relay systems. However, amplifiers consume additional power and generate noise~\cite{zhang2022active}. \Acp{hris} are \ac{ris} architectures having passive phase shifters along with a few active elements that offer a power consumption and noise trade-off~\cite {sankar2022hybrid_RIS_JRC,nguyen2022hybrid_ris_comm,kang2022hybrid_active_passive,xie2022hybrid_mec}. By appropriately designing a \ac{hris}, performance comparable to that of a \ac{aris} can be achieved while consuming lesser power. 

Typical design in \ac{hris}-assisted communication systems involves a joint design of transmit precoder and the \ac{hris} parameters to maximize a utility function such as the \ac{snr} and is often tackled using alternating optimization.  In~\cite{nguyen2022hybrid_ris_comm}, the transmit precoder and  \ac{hris} parameters are designed to maximize the sum-rate of a multi-user \ac{miso} system. In~\cite{kang2022hybrid_active_passive}, \ac{hris} parameters and the number of active and passive elements are designed to maximize \ac{snr}. In~\cite{xie2022hybrid_mec}, a more general \ac{hris} architecture wherein each element can be dynamically switched between active and passive modes is considered. While designing \acp{hris}, an interesting question is, does the placement of active and passive elements affect performance of \ac{hris}-assisted communication systems? 

In this paper, we show that by optimally placing the active elements in \acp{hris}, we can get an \ac{snr} improvement whenever the channel gains associated with each \ac{hris} element are different. To this end, we focus on channel-aware placement of active elements, which has not been studied in~\cite{nguyen2022hybrid_ris_comm,kang2022hybrid_active_passive,xie2022hybrid_mec}. We consider a dynamic \ac{hris} architecture as in~\cite{xie2022hybrid_mec} and jointly design the transmit precoder, \ac{hris} parameters, and placement of active~(and passive) elements in a channel-aware manner to maximize the received \ac{snr}. While the problem of \ac{hris} coefficient design itself is non-convex, the problem of  optimal placement is combinatorial and NP-hard. To tackle the non-convex optimization problem, we propose a low-complexity solver based on alternating optimization.  Specifically, we maximize a lower bound on the \ac{snr} and present closed-form expressions for the precoder, \ac{hris} coefficients, and a simple near-optimal heuristic for the placement of active elements. The proposed algorithm scales log-linearly with the number of \ac{hris} elements, making it also suitable for systems that require large \ac{ris} arrays. We also derive the suboptimality gap between the proposed and optimal solutions. Through numerical simulations, we demonstrate that the proposed method offers higher \ac{snr} when compared to communication systems with \ac{pris}, \ac{hris} with arbitrary active element placements, and without \ac{ris}.

\section{Problem Modeling} \label{sec:sys:model}
In this section, we present the system model and state the problem. 

\subsection{System Model}

\begin{figure}
\centering
	\includegraphics[width=0.8\columnwidth]{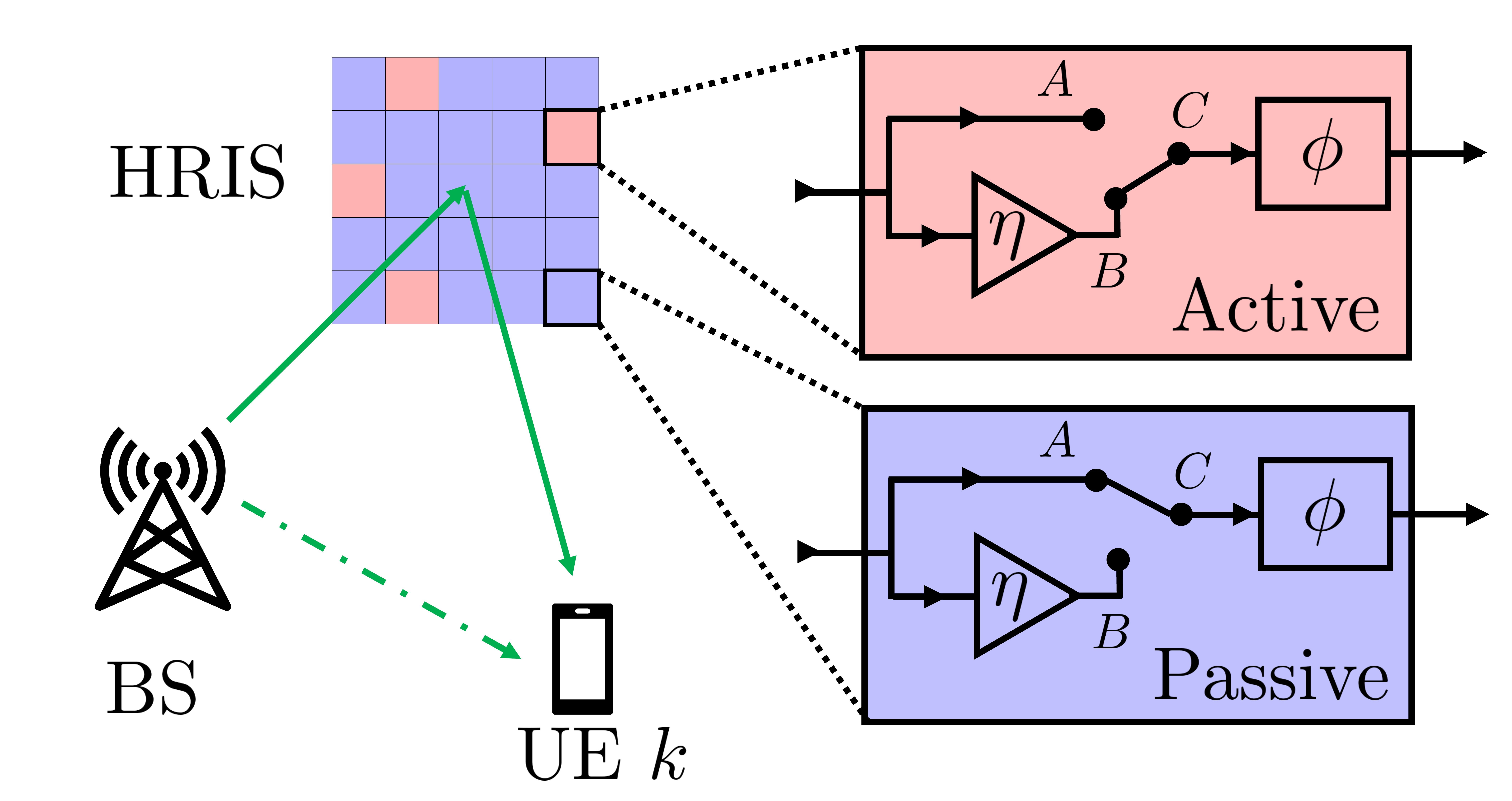}
	\caption{\footnotesize The \ac{hris} architecture. Each element is equipped with an amplifier and can be programmed to operate in passive (switch set to position AC) or active~(switch set to position BC) configurations.}
	\label{fig:sys:model}
\vspace{-4mm}
\end{figure}
Consider a \ac{hris}-assisted \ac{miso} communication system with a multi-antenna \ac{bs} serving $K$ single-antenna \acp{ue} as illustrated in~Figure~\ref{fig:sys:model}.  
 The data stream corresponding to the $k$th \ac{ue} denoted by $s_k$ is precoded by the unit norm precoder $\vp_k \in \mbC^{M}$ before transmission. 
 We assume that the \acp{ue} are served using \ac{tdm} multiple access, and henceforth, we drop the user index $k$. The discrete-time complex baseband downlink signal is given by $\vx = \sqrt{P_t}\vp s$, where $P_t$ is the total transmit power.

 Consider a \ac{hris} with $N$  elements, out of which $L$ elements are active. Let $\cA$ and $\cA^c$ denote the index set indicating the placement of active and passive elements, respectively.  Let $\omega_i \in \mbC$ be the coefficient of $i$th element in the \ac{hris}. We have $\vert \omega_i \vert \leq \eta$ for $i\in \cA$ and $\vert \omega_i \vert = 1$ for  $i \in \cA^c$, where $\eta$ is the maximum amplification of active elements.  Reflection-type amplifiers consume additional power and lead to noisy reflected signals. The  total power consumed by the  \ac{hris} is given by~\cite{nguyen2022hybrid_ris_comm}
 \begin{eqnarray}
 	P_{\rm ris} = \sum_{i \in \cA} \vert \omega_i \vert^2  \left( \mbE \left[ \vu_{\rm in} \vu_{\rm in}\rH \right] + \nu^2  \right),
 \end{eqnarray}
where $\vu_{\rm in}$ is the signal incident on the \ac{hris}.  The \ac{hris} coefficients should be designed to satisfy the total power constraint $P_{\rm ris} \leq P_{\rm ris, max}$.  Let $\mH_{\rm br} \in \mbC^{N \times M}$ be the \ac{mimo} channel between the \ac{bs} and  \ac{hris}. Hence, we have $\vu_{\rm in} = \mH_{\rm br}\vx$. Let $\mJ_{\cA} \in \{0,1\}^{N \times N}$ denote the selection matrix which selects the rows indexed by $\cA$.  Then, the noisy signal reflected from the \ac{hris} is
\begin{equation}
	\vu_{\rm out} = {\rm diag}\left( \boldsymbol{\omega} \right) \vu_{\rm in} + \mJ_{\cA} {\rm diag}\left( \boldsymbol{\omega} \right) \vq,
\end{equation} 
where $\vq \sim \cC\cN(\boldsymbol{0},\nu^2\mI)$ is the \emph{\ac{ris} noise} vector and $\boldsymbol{\omega} = [\omega_1,\ldots,\omega_N]\rT \in \mbC^{N}$ is the \ac{hris} parameter vector. The sets $\cA$ and $\cA^c$ are not fixed during fabrication and can be changed dynamically.

Let $\vh_{{\rm bu}}\rH \in \mbC^{1\times M}$ and $\vh_{{\rm ru}}\rH \in \mbC^{1 \times N}$ be the channels between the \ac{bs}  and \ac{hris}
to the \ac{ue}, respectively.  Let $\vh \rH = \vh_{{\rm bu}}\rH + \vh_{{\rm ru}}\rH{\rm diag}(\boldsymbol{\omega})\mH_{\rm br}$ denote the overall channel from the \ac{bs} to the \ac{ue}. The received signal is given by
\begin{equation}
	y = \vh\rH \vx + \vh_{{\rm ru}}\rH \mJ_{\cA} {\rm diag}(\boldsymbol{\omega}) \vq  +  n,
\end{equation}
where $\vh_{{\rm bu}}\rH \mJ_{\rm a} {\rm diag}(\boldsymbol{\omega}) \vq$ is due to RIS noise from active elements and  $n \sim \cC \cN(0,\sigma^2)$ is the receiver noise. The received \ac{snr} is given by
\begin{equation} \label{eq:def:gamma}
	\gamma(\vp,\boldsymbol{\omega},\cA) = \frac{ P_t \vert \vh \rH \vp \vert^2  }{ r(\boldsymbol{\omega},\cA)  + \sigma^2},
\end{equation}
where $r(\boldsymbol{\omega},\cA) = \nu^2 \Vert \vh_{{\rm ru}}\rH\mJ_{\cA}{\rm diag}(\boldsymbol{\omega})\Vert_2^2$ is the variance of noise due to active elements at the \ac{ue}. The corresponding spectral efficiency~(SE) is $\log_2(1+\gamma)$. The \ac{snr} is  a function of the precoder $\vp$ and the \ac{hris} coefficients $\boldsymbol{\omega}$, and also is a function of $\cA$ as we can observe from~\eqref{eq:def:gamma}. Hence, it is important to optimally design $\cA$ to improve \ac{snr}.
\subsection{Problem formulation} \label{sec:prob:form}
Data rate or spectral efficiency of a \ac{miso} communication system is determined by the received \ac{snr} $\gamma$. Hence, we propose to design the transmit precoder, $\vp$, the \ac{hris} coefficients, $\boldsymbol{\omega}$, and the placement of active elements, $\cA$, to maximize the received \ac{snr} as 
\begin{subequations} 
	\begin{align} 
	\hspace{-2mm}	(\cP) \quad \hspace{3mm} \medspace	\underset{\vp,\boldsymbol{\omega}, \cA}{\text{maximize}}  & \quad \gamma(\vp,\boldsymbol{\omega},\cA) \nonumber  \\ 
		\text{s. to }	&  \quad  \| \vp \|_2 \leq 1, \label{eq:prob:form:con:txpow} \\
		 & \quad P_{\rm ris} \leq P_{\rm ris, max}, \label{eq:prob:form:con:rispow} \\
			& \quad \vert \omega_i \vert \leq \eta,  i \in \cA,  \,  \vert \omega_i \vert = 1, i \in \cA^c, \label{eq:prob:form:con:RIS}			
	\end{align} 
\end{subequations}
where \eqref{eq:prob:form:con:txpow} is the transmit power constraint at the \ac{bs},~\eqref{eq:prob:form:con:rispow} is the \ac{ris} power constraint due to the active elements, and~\eqref{eq:prob:form:con:RIS} represents the constraints related to the active and passive elements. Due to the coupling between the optimization variables and the nonconvexity of the objective function and the constraints, $(\cP)$ is a nonconvex optimization problem. Before we develop a solver for  $(\cP)$, we make the following remark about \eqref{eq:prob:form:con:rispow}.
 \begin{myrem}
	Since $\vert \omega_i \vert \leq \eta$ for $i \in \cA$, the maximum power consumed by the \ac{hris} is upper bounded, i.e., $P_{\rm ris} \leq P_{\rm ris, ub}$~\cite[Theorem 1]{sankar2022hybrid_RIS_JRC}. Hence, if $P_{\rm ris,max} \geq P_{\rm ris, ub}$, then the \ac{hris} total power constraint is inactive. Hence, for the sake of simplicity, we ignore the \ac{hris} power constraints  in the next session. 
\end{myrem}

\section{Proposed solvers} \label{sec:altopt:baseline}
In this section, we propose an alternating optimization based solver to obtain the transmit precoder and the \ac{ris} coefficients. For a given choice of \ac{ris} coefficients $\boldsymbol{\omega}$ and $\cA$, the design of the transmit precoder to maximize $\gamma$ reduces to that of maximizing the received signal power at the \ac{ue}, for which the optimally solution is given by the maximal ratio transmission with  $\vp^\star = \vh/\Vert \vh \Vert$~\cite{tse2005fundamentals}. Hence, our main focus will be on the design of $\boldsymbol{\omega}$ and $\cA$ for a given precoder $\vp^\star$. First, we present an exhaustive search-based scheme to find the optimal placement $\cA$ and to design $\boldsymbol{\omega}$.  Next, we relax this problem and provide closed-form expressions to obtain a nearly optimal solution to the joint design of $\boldsymbol{\omega}$ and $\cA$.

\subsection{Exhaustive search-based design of $\boldsymbol{\omega}$ and $\cA$, given $\vp$} \label{sec:ex:search}

To find the optimal $\cA$ and $\boldsymbol{\omega}$ for a given $\vp^\star$, we could search through all the possible index sets $\cA$. For an $N$ element \ac{hris} with $L$ active elements, there are $\psi = {N \choose L}$ different possible index sets $\cA$.  Let $\cA_j$ for $j=1,\ldots, \psi$ be the candidate placement sets. Specifically, we design $\boldsymbol{\omega}(\cA_j)$ for $j=1,\ldots, \psi$ and pick the one resulting in the highest \ac{snr} as a solution. Given $\vp$ and $\cA_j$, the design of $\boldsymbol{\omega}$ can be performed using a \ac{sdp} as described next. Let $f = \vh_{{\rm bu}}\rH \vp$ and  $\vg = {\rm diag}\left(  \vh_{{\rm ru}} \right) \mH_{\rm br} \vp  =[g_1,\ldots,g_N] \in \mbC^{N}$ denote the effective channels corresponding to the direct path and the path via the \ac{hris}, respectively. By introducing $\vv\rH = [\boldsymbol{\omega}\rT(\cA_j),1] \in \mbC^{1\times (N+1)}$ and $\va = [\vg\rT,f]\rT \in \mbC^{N+1}$,  the subproblem w.r.t. $\boldsymbol{\omega}$ can be written as 

\begin{subequations} \label{eq:sdp:ris} 
	\begin{align} 
		\underset{\mV = \vv \vv\rH}{\text{maximize}}  & \quad \frac{ {\rm Tr}\left( \va\va\rH \mV \right) }{{\rm Tr}\left( \mB \mV \right) + \sigma^2} \nonumber \\ 
	\text{s. to}	&	\quad \vert \left[\mV \right]_{i,i} \vert \leq \eta^2,  i \in \cA_j,  \,  \vert \left[\mV \right]_{i,i}\vert = 1, i \in \cA_j^c,  \nonumber
	\end{align} 
\end{subequations}
where \[\mB = \begin{bmatrix} {\rm diag}\left( \vh_{{\rm ru}}^* \odot \vh_{{\rm ru}} \right)\mJ_{\cA} & \boldsymbol{0} \\ \boldsymbol{0}\rT & 0 \end{bmatrix}\] and  $\odot$ is the elementwise Hadamard product. Replacing the nonconvex unit-rank constraint $\mV = \vv \vv\rH$ with a relaxed semi-definite constraint $\mV \succcurlyeq \boldsymbol{0}$, we get an \ac{sdp} that can be solved using off-the-shelf solvers. To obtain $\boldsymbol{\omega}(\cA_j)$, we use Gaussian randomization~\cite{sankar2022hybrid_RIS_JRC}. Finally, $\boldsymbol{\omega}$ is computed by solving 
$\underset{j}{\text{max}} \quad \gamma(\vp^\star, \boldsymbol{\omega}(\cA_j),\cA_j).$

Obtaining a $\zeta$-accurate solution with \ac{sdp} costs about $\cO\left( N^{6.5} \log (1/\zeta) \right)$ flops. Thus, the overall complexity of designing $\boldsymbol{\omega}$ and $\cA$ using the above exhaustive search-based scheme is about $\cO\left(  \psi N^{6.5} \log (1/\zeta) \right)$, which is prohibitively high even for moderate values of $N$. Therefore, in the next section, we provide a low-complexity method, which is provably nearly optimal.

\subsection{Low-complexity design of $\boldsymbol{\omega}$ and $\cA$, given $\vp$} \label{sec:prop:solver}

Given the precoder $\vp^\star$, the subproblem of designing the \ac{hris} coefficients and the placement of active elements is given by
	\begin{subequations} 
	\begin{align} 
			\underset{\boldsymbol{\omega}, \cA}{\text{maximize}} \quad & \quad \frac{ P_{\rm t}\vert \vh_{{\rm ru}}\rH {\rm diag}\left( \boldsymbol{\omega} \right)\mH_{\rm br}\vp^\star + \vh_{{\rm bu}} \rH \vp^\star \vert^2  }{ r(\boldsymbol{\omega},\cA)  + \sigma^2} \nonumber  \\ 
		\text{s. to}	&	\quad \vert \omega_i \vert \leq \eta,  i \in \cA,  \,  \vert \omega_i \vert = 1, i \in \cA^c. \nonumber 
	\end{align} 
\end{subequations}

 To solve the above nonconvex joint optimization problem, we begin by first replacing $r(\boldsymbol{\omega},\cA)$ in the denominator with an upper bound presented in the following lemma.
\begin{mylem}
	Due to the limited amplification of active elements, the RIS noise variance at the \ac{ue} is upper bounded as 
	\begin{equation}
		r(\boldsymbol{\omega},\cA)\leq \nu^2 \eta^2  \sum_{i=1}^{L} \left\vert \left[ \vh_{{\rm ru}} \right]_{j_i} \right\vert^2 \overset{\Delta}{=} r_{\rm max},
	\end{equation}
where $j_i$ is the index corresponding to the rank ordering  $\vert \left[ \vh_{{\rm ru}} \right]_{j_1} \vert \geq \cdots \geq \vert \left[ \vh_{{\rm ru}} \right]_{j_N} \vert$.
\end{mylem}
\begin{proof}
	The \ac{ris} noise variance at the \ac{ue} is given by 
	\begin{equation}
		r(\boldsymbol{\omega},\cA) = \nu^2 \sum_{i \in \cA}  \left\vert  \left[ \vh_{{\rm ru}} \right]_i^* \omega_i \right\vert^2   \overset{(a)}{\leq} \nu^2 \eta^2 \sum_{i =1}^{L}  \left \vert \left[ \vh_{{\rm ru}} \right]_{j_i} \right \vert^2, \nonumber
	\end{equation}	
where $(a)$ follows due to the constraint that $\vert \omega_i \vert \leq \eta$ for $i\in \cA$. Since the number of active elements is $L$, the upper bound in $(a)$ follows from the definition of the index $j_i$. 
\end{proof}
Using the above upper bound on $r(\boldsymbol{\omega},\cA)$, we can lower bound the \ac{snr} as
\begin{equation} \label{eq:gamma:min:def}
	\gamma \geq  \frac{ P_{\rm t} \vert \vh_{{\rm ru}}\rH {\rm diag}\left( \boldsymbol{\omega} \right)\mH_{\rm br}\vp^\star + \vh_{{\rm bu}} \rH \vp^\star \vert^2  }{ r_{\rm max}  + \sigma^2} \overset{\Delta}{=} \gamma_{\rm min}.
\end{equation}
Now, we design the \ac{ris} coefficients and the placement of active elements to maximize the \ac{snr} lower bound $\gamma_{\rm min}$. As in the transmit precoder design,  maximizing $\gamma_{\rm min}$ is equivalent to maximizing the received signal power at the \ac{ue}. However, the problem of choosing the active element placement is still nonconvex and combinatorial. Interestingly, for the considered scenario, we have a simple solution, which we present as the following theorem. 

\begin{theorem} \label{theo:main}
The optimal active element placement that maximizes $\gamma_{\rm min}$ is given by
	\begin{equation} \label{eq:theo:opt:placement}
		\cA^\star = \{a_1,a_2,\ldots,a_L\},
	\end{equation} 
where the indices $a_i$ are defined according to the rank ordering $\vert g_{a_1} \vert \geq  \ldots \geq \vert g_{a_N} \vert$. Furthermore, the optimal RIS coefficients is given by
\begin{equation} \label{eq:theo:opt:coeff}
	\omega_i^\star = \begin{cases}
	\eta \exp\{-\jmath \left( \theta\left( g_i \right) - \theta\left( f \right) \right) \},\quad &i \in \cA, \\ \exp\{-\jmath \left( \theta\left( g_i \right) - \theta\left( f \right) \right) \},\quad &i \in \cA^c, 
	\end{cases}
\end{equation}
where $\theta\left(\cdot\right)$ denotes the phase of the complex number.
\end{theorem}

\begin{proof}
Let us call the numerator of \ac{snr} as $\vert c(\boldsymbol{\omega},\cA)\vert^2$ with $	c(\cA, \boldsymbol{\omega}) = f + \boldsymbol{\omega}\rT\vg =  \vert f \vert e^{\jmath\theta \left( f \right)} + \sum_{i=1}^{N}\omega_i \vert g_i \vert e^{\jmath \theta \left(g_i\right)}$.
We further have
 \begin{align}
 	\vert c(\cA,\boldsymbol{\omega}) \vert & \leq \vert f \vert + \sum_{i \in \cA} \vert \omega_i g_i \vert + \sum_{i \in \cA^c} \vert \omega_i g_i \vert \nonumber\\
 	& \overset{(a)}{\leq} \vert f \vert +  \sum_{i \in \cA} \eta \vert g_i \vert + \sum_{i \in \cA^c} \vert g_i \vert ,  \label{eq:num:max}
 \end{align}
where $(a)$ follows from the amplification constraint~\eqref{eq:prob:form:con:RIS}. Using~\eqref{eq:theo:opt:coeff}, we have  
\begin{align}
		\vert c(\cA,\boldsymbol{\omega}^\star) \vert &= \left\vert \vert f \vert e^{\jmath \theta \left(f \right)} + \sum_{i \in \cA} \eta e^{-\jmath \left( \theta \left( g_i \right) - \theta \left( f \right) \right)} \vert g_i \vert e^{\jmath \theta \left( g_i \right)} \right.  \nonumber \\
		&+ \left. \sum_{i \in \cA^c}  e^{-\jmath \left( \theta \left( g_i \right) - \theta \left( f \right) \right)}  \vert g_i \vert e^{\jmath \theta \left( g_i \right) } \right\vert \nonumber \\
		& =\left\vert  \vert f \vert + \left( \sum_{i \in \cA} \eta \vert g_i \vert + \sum_{i \in \cA^c} \vert g_i \vert  \right)  \right\vert. 
\end{align}

Hence, by choosing $\boldsymbol{\omega}$ as in~\eqref{eq:theo:opt:coeff}, $c(\cA,\boldsymbol{\omega}^\star)$ achieves the upper bound given in~\eqref{eq:num:max} with equality, thereby proving that the optimal solution is indeed given by~\eqref{eq:theo:opt:coeff}. We now prove that the optimal active element placement is given by~\eqref{eq:theo:opt:placement}. For any given placement of active elements, selecting the \ac{hris} coefficients using~\eqref{eq:theo:opt:coeff} ensures that $\vert c(\cA,\boldsymbol{\omega}^\star) \vert$ attains the maximum possible value. Therefore, to choose the optimal placement $\cA^\star$, we maximize $\vert c(\cA,\boldsymbol{\omega}^\star)\vert$ w.r.t. $\cA$. From the definition of the indices $a_i$, it follows that
\begin{equation} \label{eq:elem:place:ub}
\vert	c(\cA,\boldsymbol{\omega}^\star) \vert \leq \vert f \vert + \sum_{i=1}^{L} \eta \vert g_{a_i} \vert + \sum_{i=L+1}^{N} \vert g_{a_i} \vert .
\end{equation} 
 We can now observe that $\vert c(\cA,\boldsymbol{\omega}^\star) \vert$ attains the upper bound with equality for the placement in~\eqref{eq:theo:opt:placement}. 
\end{proof}

We summarize the proposed solver as Algorithm~\ref{alg:main}. The major complexity in obtaining $\cA^\star$ and $\boldsymbol{\omega}^\star$ stems from the sorting operation needed to obtain $\cA^\star$. An $N$-long vector can be efficiently sorted at a complexity of the order of $\cO\left( N \log N \right)$~\cite{cormen2009algorithms},  making the proposed algorithm several orders computationally less complex than the exhaustive search-based approach presented in Sec.~\ref{sec:ex:search}.  In essence, the complexity of proposed method  scales log-linearly as a function of $N$ making the algorithm suitable for large \acp{ris}. Before we end this section, we make the following two remarks.

\begin{myrem} \label{corr:optim:place}
Optimal placement of active elements is necessary to fully leverage a \ac{hris} unless both $\mH_{\rm br}$ and $\vh_{{\rm ru}}$ are \ac{los} channels with $\vert g_{j_1} \vert = \ldots = \vert g_{j_N} \vert$ for which the upper bound in~\eqref{eq:elem:place:ub} is achieved with equality for any $\cA$. In other words, for fading channels, an optimal placement is necessary to maximize the received signal power at the \ac{ue}.
\end{myrem}
 
 \begin{myrem} 
 	If $P_{\rm ris}(\boldsymbol{\omega}^\star,\cA^\star) > P_{\rm ris, max}$, i.e., when the \ac{ris} power constraint is active, we can obtain a solution satisfying the \ac{ris} power constraints by scaling the co-efficients of the active elements as
 	\begin{equation} \label{eq:pow:con}
 		\tilde{\omega}_i = \left( \frac{P_{\rm ris,max}}{ \sum_{i \in \cA} \vert \omega_i^\star \vert^2  \left( \mbE \left[ \vu_{\rm in} \vu_{\rm in}\rH \right] + \nu^2  \right)}  \right)^{1/2} {\omega}_i^\star, \hspace{2.5mm} i \in \cA.
 	\end{equation}	
 \end{myrem}

\begin{algorithm}[!t] 
	\caption{Proposed solver for $\vp$, $\boldsymbol{\omega}$, and $\cA$.}\label{alg:main}
	{\bf Initialization:}  $\boldsymbol{\omega}^{(n)} = [\eta,0,\ldots,0]\rT$, $\cA^{(n)} = \{1,\ldots,L\} $ 
	\begin{algorithmic}[1]
		\For{$n=1, 2, \cdots,\texttt{MaxIter}$}
		\State Set $\vp^\star = \vh/\Vert \vh \Vert$, $\vh\rH = \vh_{{\rm ru}}\rH{\rm diag}\left( \boldsymbol{\omega}^{(n)}\right)\mH_{\rm br} + \vh_{{\rm bu}}\rH$
		\State Set $\cA^{(n)}$ using~\eqref{eq:theo:opt:placement} and $\boldsymbol{\omega}^{(n)}$ using~\eqref{eq:theo:opt:coeff}
		\EndFor	
	\end{algorithmic}
\end{algorithm}

\section{Theoretical analysis} \label{sec:analysis}

In this section, we theoretically show that the performance gap between the solution obtained using the exhaustive search-based scheme and the proposed solution is upper bounded when $\eta$ is appropriately selected. Let $\gamma^{\rm opt}$ be the optimal \ac{snr} obtained using the exhaustive search-based method. Let $\gamma^{\rm prop}$ be the \ac{snr} of the proposed solver. Substituting the obtained solution in~\eqref{eq:def:gamma}, we get 
	\begin{equation}
		\gamma^{\rm prop} = \frac{ c^{\rm max} }{ \sum_{i \in \cA} \nu^2 \eta^2  \vert  \left[ \vh_{{\rm ru}} \right]_{i}  \vert^2 + \sigma^2 },
	\end{equation}
	where $c^{\rm max} = \left \vert \vert f \vert + \sum_{i=1}^{L} \eta \vert g_{a_i} \vert  + \sum_{i=L+1}^{N} \vert g_{a_i} \vert \right \vert^2 $. 

 \begin{theorem} \label{theo:boundedness}
	Let us define $E \overset{\Delta}{=} \frac{ \sigma^2 }{ c^{\rm max}} \vert \gamma^{\rm opt} - \gamma^{\rm prop} \vert$ as the normalized suboptimality gap. Then, $E \leq \delta$ if the amplification factor satisfies 
	\begin{equation} \label{eq:theo3:main}
		\eta \leq \sqrt{  \frac{ \delta \sigma^2 }{ (1-\delta)\nu^2 \sum_{i=1}^{L} \left \vert \left[\vh_{{\rm ru}}\right]_{j_i} \right\vert^2  } } \overset{\Delta}{=} \eta_{\rm max}.
	\end{equation}
\end{theorem}
The proof is given in the appendix. From~\eqref{eq:def:gamma} and~\eqref{eq:gamma:min:def}, we can observe that the difference between $\gamma_{\rm min}$ and $\gamma$ will be small for small values of $r(\boldsymbol{\omega},\cA)$. Since the \ac{ris} noise variance $\nu^2$ and the channel gains are not under our control, choosing a small enough $\eta$ will ensure that $r(\boldsymbol{\omega},\cA)$ (and hence, the performance gap) is small. Theorem~\ref{theo:boundedness} allows us to  choose an appropriate value of $\eta$ so that the performance gap is bounded. Through numerical simulations, we show that for typical values of the system parameters, the values of $\gamma^{\rm prop}$ and $\gamma^{\rm opt}$ are comparable.

 	\begin{figure*}[ht]
  	\begin{subfigure}[c]{0.65\columnwidth}\centering
 		\includegraphics[width=\columnwidth]{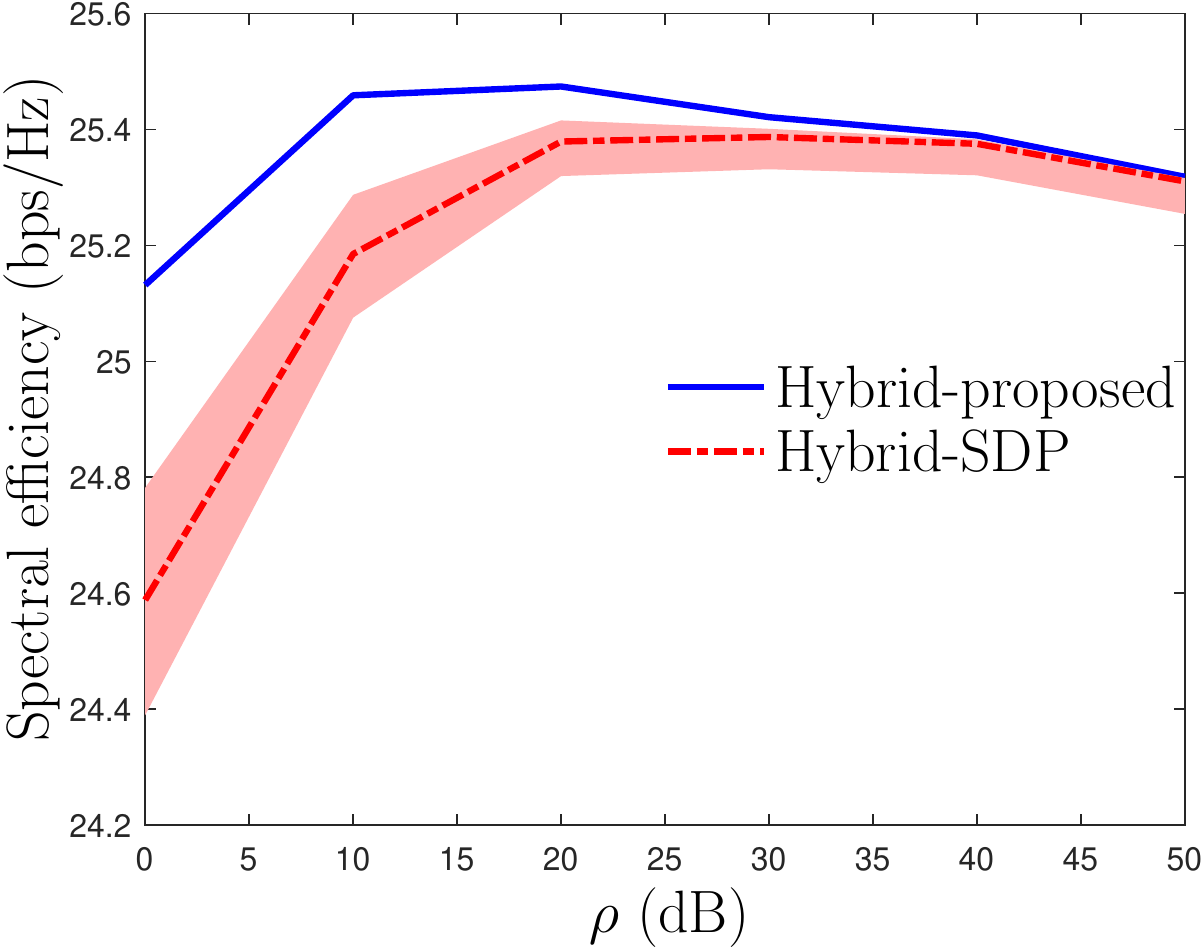}
 		\caption{}
 		
 	\end{subfigure}
 	~
 	\begin{subfigure}[c]{0.65\columnwidth}\centering
 		\includegraphics[width=\columnwidth]{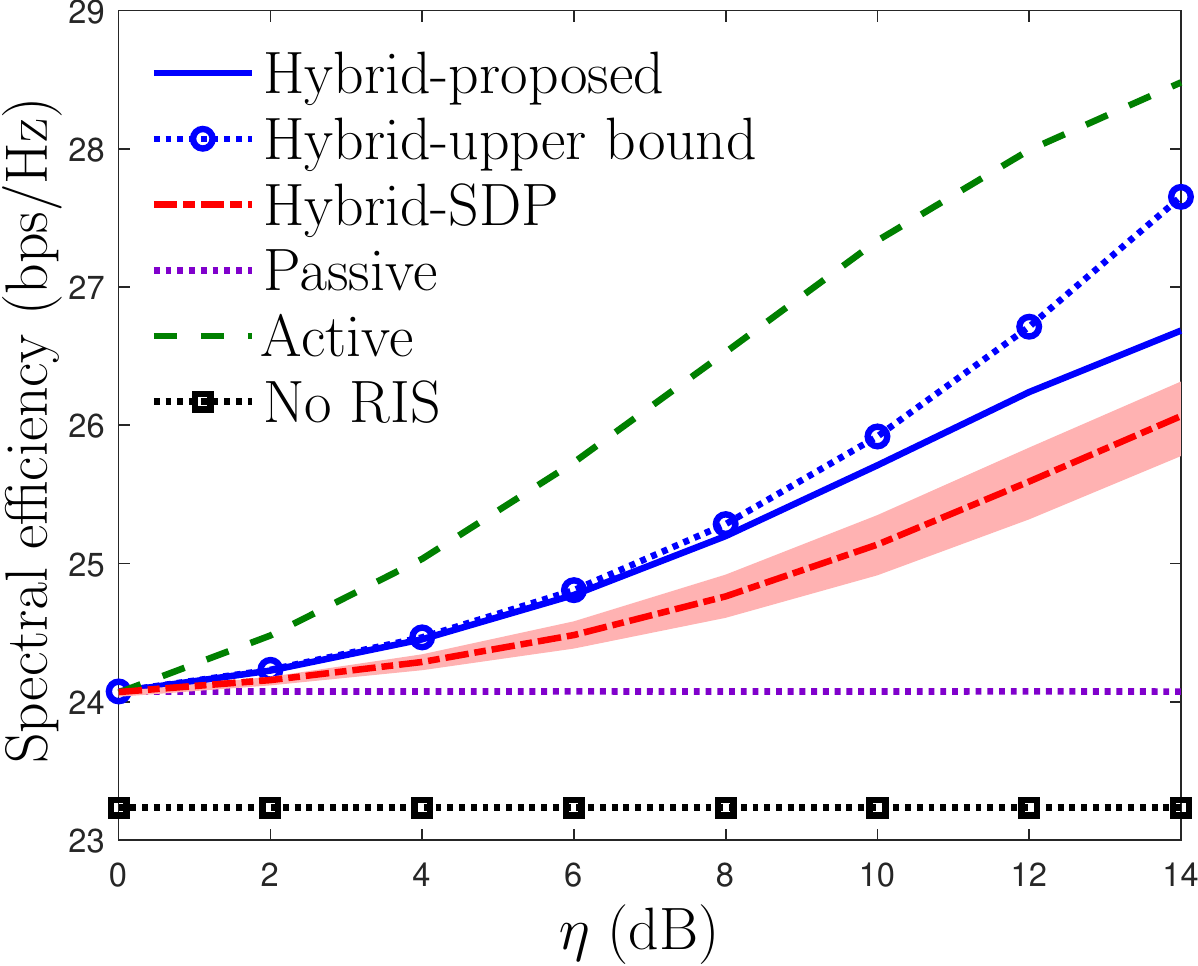}
 		\caption{}
 		
 	\end{subfigure}
 	~
 	\begin{subfigure}[c]{0.65\columnwidth}\centering
 		\includegraphics[width=\columnwidth]{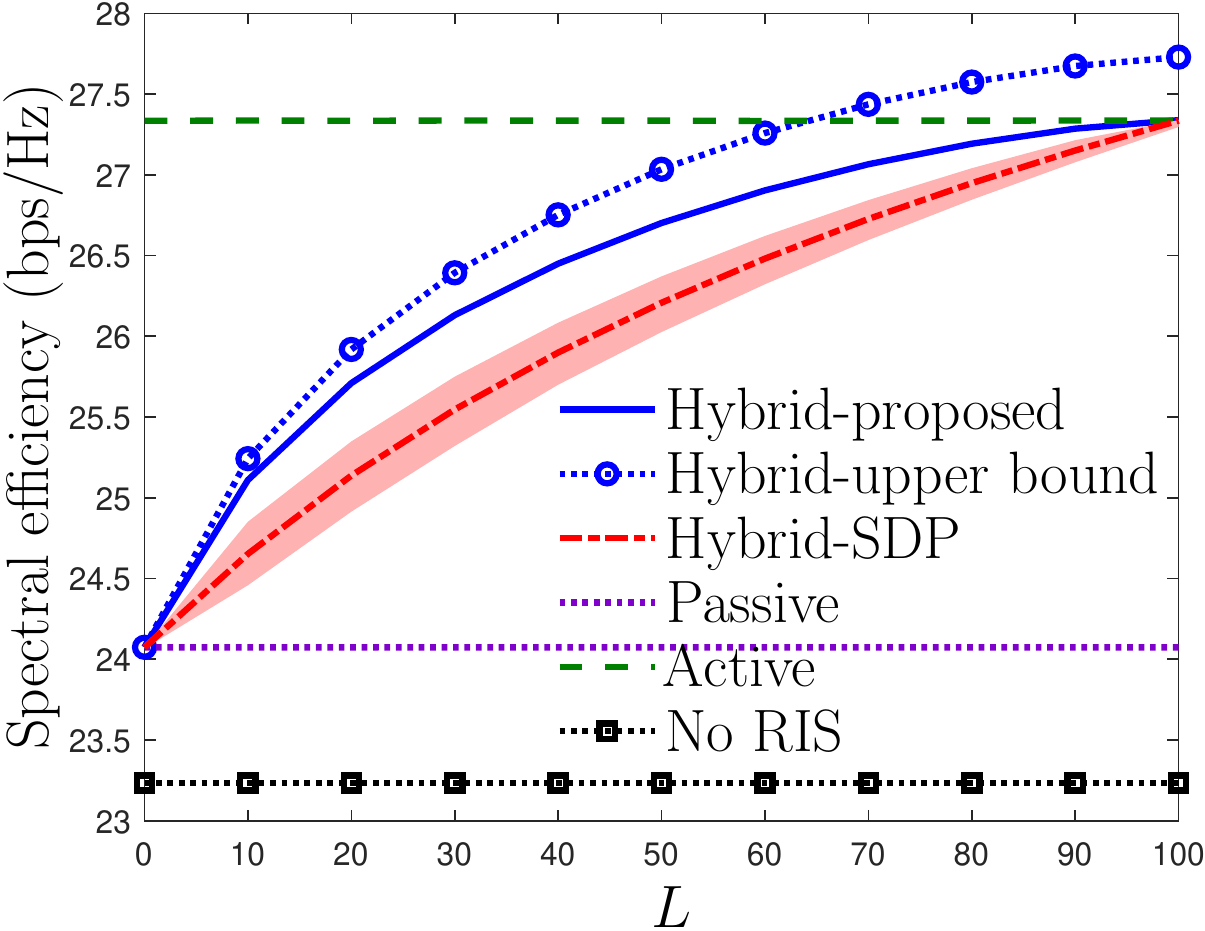}
 		\caption{}
 		
 	\end{subfigure}
 	\caption{\small (a) Impact of $\rho$. (b) Impact of $\eta$. (c) Impact of $L$.}
 	\label{fig:sim}
  \vspace{-4mm}
 \end{figure*}

\section{Simulations} \label{sec:simulations}

In this section, we present numerical simulations demonstrating the advantage of the proposed algorithm. The \ac{bs} and \ac{hris} are assumed to be located at $(0,0,0)$m and $(20,13,3)$m, respectively. The user location is drawn randomly from a $3\times 10$ rectangular area with the bottom left corner at  $(18,8,0)$m. All wireless links are assumed to follow Rician distribution with a Rician factor $\rho$ and a pathloss of $30+22 {\rm log}(d)$ dB, where $d$ is the distance between the terminals in meter.

The \ac{bs} is equipped with $M=8$ antennas. Unless mentioned otherwise, the \ac{ris} comprises $N=100$ elements, out of which $L = 20$ elements are active with a maximum amplification factor of $\eta = 10$~dB. The total transmit power constraint at the \ac{bs} is $P_{\rm t} = 10$~dB and at the \ac{ris} is $P_{\rm ris, max} = 0$ dB. The noise variances are set to $\sigma^2 = -80$~dBm and $\nu^2 = -80$~dBm. We   observe that the \ac{ris} power constraint is inactive for the considered setting.

We compare the performance of the proposed method, referred to as \texttt{Hybrid-proposed}, with that of designing the precoder and \ac{hris} coefficients using \ac{sdp} for arbitrary placement of active elements, referred to as \texttt{Hybrid-SDP}. Specifically, to simulate \texttt{Hybrid-SDP}, we design $\vp$ and $\boldsymbol{\omega}$ using the \ac{sdp}-based  method presented in Sec.~\ref{sec:altopt:baseline} for $100$ arbitrary candidate active element locations (instead of ${100 \choose 20}$). By the dotted-dash line in Figure~\ref{fig:sim}, we indicate the average value of \texttt{Hybrid-SDP} (averaged over the $100$ arbitrary placements). The region between the best and worst placement is  shaded.
We also compare the proposed method with that of using a fully passive \ac{ris} (\texttt{Passive}), a fully active \ac{ris} (\texttt{Active}), and a system without \ac{ris} (\texttt{No RIS}).  All subsequent plots are obtained by averaging over $100$ independent channel realizations.

 We present the performance of \texttt{Hybrid-proposed} and \texttt{Hybrid-SDP} for different Rician factors in Fig.~\ref{fig:sim}(a). For lower values of $\rho$, we can observe that an optimally designed \ac{hris} has a better performance when compared to a scheme with arbitrary placement. For large Rician factors, the performance gap is indeed decreasing due to the \ac{los} nature of the associated channels wherein $\vert [\vg]_{1} \vert \approx \vert [\vg]_{2} \vert \approx \ldots \approx \vert [\vg]_{N} \vert$. For subsequent simulations, we fix $\rho = 10$ for the direct link $\vh_{{\rm bu}}$ and the \ac{bs}-\ac{ris} link $\mH_{\rm br}$. The \ac{ris}-\ac{ue} link  is assumed to be Rayleigh  distributed (i.e., $\rho = 0$).

The impact of the amplification factor $\eta$ and the number of active elements $L$ on the performance of the proposed scheme is presented in Figures~\ref{fig:sim}(b) and~\ref{fig:sim}(c), respectively. The proposed scheme is significantly better than that of \texttt{Passive} and \texttt{No RIS}. Moreover, throughout the considered range, the performance of low-complexity \texttt{Hybrid-proposed} is better than that of \texttt{Hybrid-SDP}. Generally, the performance of methods comprising active elements increases with an increase in $\eta$  or $L$. This is expected since a higher value of $\eta$ and $L$ results in an improved array gain due to the active elements. It is also interesting to note that the proposed scheme often results in almost {\color{black} $1$ bps/Hz higher SE} (for example, at $\eta=10$ and $L = 40$) when compared with the worst arbitrary placement. With a carefully designed placement of active elements, the proposed method also achieves comparable performance (within around {\color{black} 98 $\%$ SE}) to that of a fully active \ac{ris} with just $L=60$ as opposed to a \texttt{Hybrid-SDP}, which requires $L=80$. To summarize, the numerical results clearly demonstrate that a channel-aware optimal placement is crucial to benefit from \ac{hris} architectures fully.

\section{Conclusion} \label{sec:conclusion}

In this paper, we proposed an algorithm having log-linear complexity to design the transmit precoders, \ac{ris} reflection coefficients, and the placement of the active elements to maximize the \ac{ue} \ac{snr}. To obtain a low-complexity algorithm and a tractable solution, we relaxed the problem of maximizing the \ac{snr} with that of maximizing a lower bound on the \ac{snr} and presented simple closed-form expression for the precoder and \ac{hris} coefficients. We have also theoretically showed that the proposed solution is close to the exhaustive search-based solution and derived a condition under which the performance gap is bounded. Through numerical simulations, we have demonstrated that the proposed method is significantly better than using either a fully passive RIS or a Hybrid RIS without optimal active element placement. 

\appendices
	\section{Proof of Theorem~\ref{theo:boundedness}}
	
	\begin{proof}
		We first show that the normalized performance gap $E$ is upper bounded. We then show that $E \leq \delta$ when the amplification factor satisfies~\eqref{eq:theo3:main}. We begin by computing a lower bound for $\gamma^{\rm prop}$ and an upper bound for $\gamma^{\rm opt}$. Let us recall that the indices $\{j_1,\ldots,j_N\}$ are defined so that $\vert \left[ \vh_{{\rm ru}} \right]_{j_1} \vert \geq \ldots \geq \vert \left[ \vh_{{\rm ru}} \right]_{j_N} \vert $. It follows that $\sum_{i\in \cA} \left \vert  \left[ \vh_{{\rm ru}} \right]_i  \right\vert^2  \leq \sum_{i=1}^{L} \left \vert \left[ \vh_{{\rm ru}} \right]_{j_i}  \right\vert^2  $. Hence, we have {\color{black}$	\gamma^{\rm lb} \overset{\Delta}{=}  \frac{ c^{\rm max} }{ \sigma^2(1+\epsilon)} \leq \gamma^{\rm prop}$},
		where $\epsilon = \nu^2 \eta^2 \sum_{i=1}^{L} \left \vert \left[ \vh_{{\rm ru}} \right]_{j_i}  \right\vert^2/ \sigma^2 $. Moreover, since the RIS noise term is non-negative, we can conclude that the maximum achievable SNR will be the largest when $\nu^2 = 0$. Hence, we have
	{\color{black}$\gamma^{\rm ub} \overset{\Delta}{=}  \sigma^{-2}c^{\rm max}  \geq \gamma^{\rm opt}$}		and  $\gamma^{\rm lb} \leq \gamma^{\rm prop} \leq \gamma^{\rm opt} \leq \gamma^{\rm ub}$. Now, let us compute the difference between the upper and lower bounds of the received SNR. We have
		\begin{equation} \label{eq:theo:analysis:proof1}
		E = \frac{\sigma^2}{c^{\rm max}}	\vert \gamma^{\rm ub} - \gamma^{\rm lb} \vert =    1  -  \frac{1}{1+\epsilon}   =  \frac{\epsilon}{1+\epsilon}.
		\end{equation}
		Since $ \vert \gamma^{\rm opt} - \gamma^{\rm prop} \vert  \leq  \vert \gamma^{\rm ub} - \gamma^{\rm lb} \vert  $, it follows that $E \leq \frac{\epsilon}{1+\epsilon}.$
		
		Let us recall that $\epsilon$ depends on $\eta$. Since $E$ is upper bounded by $\frac{\epsilon}{1+\epsilon}$, it is sufficient to choose $\eta$ so that $\frac{\epsilon}{1+\epsilon} \leq \delta$, which implies $\epsilon \leq \frac{\delta}{1-\delta}$. Using the definition of $\epsilon$, we finally have  {\color{black} $\epsilon = 	\sigma^{-2} \nu^2 \eta^2 \sum_{i=1}^{L}  \left\vert   \left[  \vh_{{\rm ru}} \right]_{j_i}  \right\vert^2    \leq \frac{\delta}{1-\delta} $.}
		On re-arranging and taking the square root, we arrive at~\eqref{eq:theo3:main}.
	\end{proof}

	\bibliographystyle{IEEEtran}
	\bibliography{IEEEabrv,bibliography}

\end{document}